\documentclass[12pt,a4paper]{article}

\usepackage{graphicx}
\usepackage{url}
\usepackage{a4wide}
\usepackage{tikz}
\usepackage{hhline}

\usepackage[english, ruled, onelanguage, linesnumbered, vlined]{algorithm2e}

\usepackage{todonotes}
\usepackage{amsthm, amsmath, amssymb, amsxtra}
\usepackage[T1]{fontenc}
\usepackage{stmaryrd}
\usepackage{bbm}

\graphicspath{{Figures/}}

\SetKwComment{com}{$\triangleright$\ }{}
\SetKwProg{Fn}{Function}{:}{}

\newtheorem{theorem}{Theorem}
\newtheorem{lemma}[theorem]{Lemma}
\newtheorem{definition}{Definition}
\newtheorem{proposition}{Proposition}

\newcommand{\ie}{{\em i.e.}}

\newcommand{\myvol}{\ensuremath{\mbox{\it vol}}}

\newcommand{\mydim}{\ensuremath{\mbox{\it dim}}}
\newcommand{\myDist}{\ensuremath{\mbox{\it Dist}}}

\newcommand{\myLL}{\ensuremath{\mbox{\it LL}}}

\newcommand{\myPrev}{\ensuremath{\mbox{\it Prev}}}
\newcommand{\myNext}{\ensuremath{\mbox{\it Next}}}
\newcommand{\myNone}{\ensuremath{\mbox{\it None}}}

\newcommand{\divvol}{\ensuremath{\boxslash}}
\newcommand{\plusvol}{\ensuremath{\boxplus}}
\newcommand{\minusvol}{\ensuremath{\boxminus}}
\newcommand{\timesvol}{\ensuremath{\boxdot}}
\newcommand{\bigplusvol}{\scalebox{1.2}{\plusvol}}
\newcommand{\eventtimes}{\ensuremath{\mathcal{T}}}

\newcommand{\reaches}{\ensuremath{\longrightarrow}}
\newcommand{\myR}{\ensuremath{\mbox{\it Result}}}

\newcommand*{\diff}{\mathop{}\!\mathrm{d}}
\newcommand{\sfp}{shortest fastest path}
\newcommand{\bfrac}[2]{\ensuremath{\frac{\sigma((i,#1),(j,#2),(t,v))}{\sigma((i,#1),(j,#2))}}}

\newcommand{\mysize}{\ensuremath{\mbox{\it size}}}
\newcommand{\mySP}{\ensuremath{\mbox{\it SP}}}
\newcommand{\mySFP}{\ensuremath{\mbox{\it SFP}}}
\newcommand{\latpairset}{\ensuremath{\mbox{\it LP}}}

\newcommand{\intcc}[1]{\ensuremath{ \mathopen{[}#1\mathclose{]} }}
\newcommand{\intco}[1]{\ensuremath{ \mathopen{[}#1\mathclose{[} }}
\newcommand{\intoc}[1]{\ensuremath{ \mathopen{]}#1\mathclose{]} }}
\newcommand{\intoo}[1]{\ensuremath{ \mathopen{]}#1\mathclose{[} }}

\newcommand{\Tfrac}[2]{%
  \ooalign{%
    $\genfrac{}{}{1.6pt}1{\phantom{#1}}{\phantom{#2}}$\cr%
    $\genfrac{}{}{0pt}1{#1}{#2}$\cr%
    $\color{white}\genfrac{}{}{.6pt}1{\phantom{#1}}{\phantom{#2}}$}%
}
\newcommand{\Dfrac}[2]{%
  \ooalign{%
    $\genfrac{}{}{1.8pt}0{\phantom{#1}}{\phantom{#2}}$\cr%
    $\genfrac{}{}{0pt}0{#1}{#2}$\cr%
    $\color{white}\genfrac{}{}{.4pt}0{\phantom{#1}}{\phantom{#2}}$}%
}

\DontPrintSemicolon

\begin{document}

\begin{center}
{\Large \bf
Computing Betweenness Centrality in Link Streams
}\\
\medskip
Fr\'ed\'eric Simard, Cl\'emence Magnien and Matthieu Latapy\,\footnote{Sorbonne Universit\'e, CNRS, LIP6, F-75005 Paris, France -- \url{Matthieu.Latapy@lip6.fr}}
\end{center}

\medskip
\begin{abstract}
Betweeness centrality is one of the most important concepts in graph analysis. It was recently extended to link streams, a graph generalization where links arrive over time. However, its computation raises non-trivial issues, due in particular to the fact that time is considered as continuous. We provide here the first algorithms to compute this generalized betweenness centrality, as well as several companion algorithms that have their own interest. They work in polynomial time and space, we illustrate them on typical examples, and we provide an implementation.
\end{abstract}

%\todo[inline]{utiliser les macros pour les intervalles -- partout}

%%%%%%%%%%%%%%%%%%%%%%%%%%%%%%%%%%%%%%%%
\section{Introduction}
\label{sec:introduction}

Betweenness centrality, or betweenness for short, is one of the most classical and important concepts defined over graphs and used in the field of complex networks and social network analysis \cite{Zweig:2230510,wasserman1994social,Masuda2016Guide,latora_nicosia_russo_2017,Freeman1977betweenness}. Given a graph $G = (V,E)$, it measures how frequently each node $v \in V$ is involved in shortest paths:
$
B(v) = \sum_{u\in V, w\in V} \frac{\sigma(u,w,v)}{\sigma(u,w)}
$
where $\frac{\sigma(u,w,v)}{\sigma(u,w)}$ is the fraction of all shortest paths from $u$ to $w$ that involve $v$ if there is a path from $u$ to $w$, $0$ otherwise. Reference algorithms compute the betweenness of all nodes in a graph in time $O(n\cdot m)$, where $n$ and $m$ are the number of nodes and links in the graph \cite{Brandes2001Faster}.

Betweenness was extended recently to link streams \cite{DBLP:journals/snam/LatapyVM18}, a family of formal objects that model sequences of interactions over time in a way similar to the modeling of relations by graphs.
They are equivalent to other objects like time-varying graphs (TVG) \cite{DBLP:journals/paapp/CasteigtsFQS12,DBLP:journals/snam/BatageljP16}, relational event models (REM) \cite{SOME:SOME203,Stadtfeld2017}, or temporal networks \cite{Nicosia2013Graph,HOLME201297}, with an emphasis on the streaming nature of link sequences. Various temporal extensions of beweenness were introduced in these contexts, see Section~\ref{sec:related}.

Betweenness in link streams has some unique features that make it quite different from other temporal extensions of betweenness in graphs. In particular, it considers continuous time and links with or without durations: nodes may be linked at specific time instants, as well as during continuous periods of time. Also, it considers paths from any node at any time instant to any node at any time instant, which induces an uncountable amount of temporal nodes. This raises specific algorithmic challenges, that we address in this paper, thus obtaining the first algorithm (and implementation) for computing betweenness centrality in link streams.

We first introduce key concepts and notations in Section~\ref{sec:preliminaries}. We then show that betweenness computations involve uncountable sets of paths with a finite volume, that we define and compute in Section~\ref{sec:volumes}. In addition, it involves integrals that must be tranformed into discrete sums over a finite number of time intervals. We define and compute these intervals in Section~\ref{sec:latencies}, and combine them in Section~\ref{sec:contribution} to obtain the contribution of any pair of nodes to the betweenness of a given temporal node. We finally compute the betweenness of any temporal node in polynomial time and space, and show results on non-trivial toy examples in Section~\ref{sec:betweenness}. We provide an open Python implementation of these algorithms \cite{btwurl}.

\section{Preliminaries}
\label{sec:preliminaries}

A link stream $L$ is a triplet $(T,V,E)$ where $T = [\alpha,\omega]$ is an interval of $\mathbb{R}$ representing time, $V$ is a finite set of nodes, and $E \subseteq T\times V\otimes V$ is the set of links\,\footnote{We make the distinction between the set $X \times Y$ of ordered pairs of elements of $X$ and $Y$, that we denote by $(x,y)$ with $x\in X$ and $y\in Y$, and the set $X \otimes Y$ of unordered pairs of distinct elements of $X$ and $Y$, that we denote by $xy$ with $x\in X$, $y\in Y$ and $x \ne y$; $(x,y) \ne (y,x)$ while $xy=yx$.}. Then, $(t,uv) \in E$ means that $u$ and $v$ are linked together at time $t$. For any $u$ and $v$ in $V$, $T_{uv} = \{t, (t,uv) \in E\}$ denotes the set of time instants at which $u$ and $v$ are linked together.
See Figure~\ref{fig:ex-intro} for an illustration and \cite{DBLP:journals/snam/LatapyVM18} for a full presentation of the formalism.

We assume here that $T_{uv}$ is the union of a finite number of disjoint closed intervals (possibly singletons) of $T$. We denote by $\eventtimes$ the set of bounds of maximal intervals in $T_{uv}$ for any $u$ and $v$, that we call event times. We denote by $\overline{m}_{uv}$ the number of maximal intervals in $T_{uv}$, and by $\overline{m} = \sum_{u,v\in V} \overline{m}_{uv}$ their sum, {\em i.e.} the number of maximal intervals in $E$. In the case of Figure~\ref{fig:ex-intro}, we obtain $\eventtimes = \{1,\allowbreak 2,\allowbreak 3,\allowbreak 5,\allowbreak 6,\allowbreak 7,\allowbreak 8,\allowbreak 9,\allowbreak 11,\allowbreak 12,\allowbreak 14,\allowbreak 15,\allowbreak 16,\allowbreak 18,\allowbreak 19,\allowbreak 22,\allowbreak 23,\allowbreak 24,\allowbreak 25,\allowbreak 27,\allowbreak 28,\allowbreak 29,\allowbreak 30,\allowbreak 31\}$, $\overline{m}_{ab} = 3$, $\overline{m}_{ac} = 1$, $\overline{m}_{bc} = 4$, $\overline{m}_{bd} = 1$, $\overline{m}_{cd} = 3$, $\overline{m}_{de} = 4$, and so $\overline{m} = 16$.

Given a link stream $L=(T,V,E)$ and a time $t$, we define the graph $G_t = (V,E_t)$ with $E_t = \{uv, (t,uv)\in E\}$. We denote by $N_t(v)$ the set of neighbors of $v$ in $G_t$. We denote by $\sigma_t(u,v)$ the (finite) number of paths from $u$ to $v$ in $G_t$, and by $d_t(u,v)$ the distance from $u$ to $v$ in this graph.

\begin{figure}[!h]
\centering
\includegraphics[width=.8\textwidth]{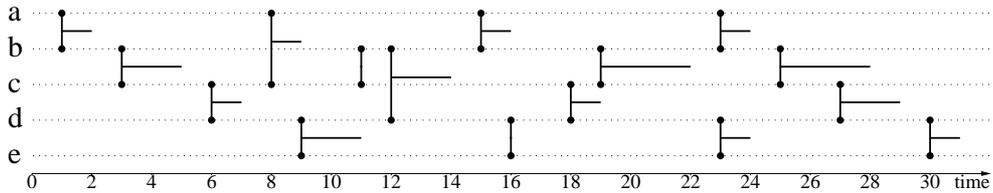}
\caption{
{\bf An example of link stream} $L=(T,V,E)$ with $T=[\alpha,\omega] = [0,32]$, $V=\{a,b,c,d,e\}$, and $E$ defined by $T_{ab} = [1,2] \cup [15,16] \cup [23,24]$, $T_{ac} = [8,9]$, $T_{bc} = [3,5] \cup \{11\} \cup [19,22] \cup [25,28]$, $T_{bd} = [12,14]$, $T_{cd} = [6,7] \cup [18,19] \cup [27,29]$, and $T_{de} = [9,11] \cup \{16\} \cup [23,24] \cup [30,31]$.
}
\label{fig:ex-intro}
\end{figure}

In $L=(T,V,E)$, a path $P$ from $(x,u) \in T\times V$ to $(y,v) \in T\times V$ is a sequence $v_0,t_1,v_1,t_2,v_2,\dots t_k,v_k$ with $v_0=u$, $v_k=v$, $x\le t_1\le t_2 \le \cdots \le t_k \le y$, and $(t_i,v_{i-1}v_i) \in E$ for all $i$. If such a path exists, then $(y,v)$ is reachable from $(x,u)$, which we denote by $(x,u) \reaches (y,v)$.
The path $P$ involves $(t_1,u)$, $(t_k,v)$, and $(t,v_i)$ for all $t \in [t_i,t_{i+1}]$ and all $i$. It starts at $t_1$, arrives at $t_k$, has length $k$ and duration $t_k-t_1$. A path with duration $0$ is called an instantaneous path.

For instance, in the case of Figure~\ref{fig:ex-intro}, the sequences $a,1,b,4,\allowbreak c,6,d,9,e$ and $a,9,c,18,\allowbreak d,27,c,28,d,30,e$ are two paths from $(0,a)$ to $(32,e)$ with length $4$ and duration $8$, and length $5$ and duration $21$, respectively.

The path $P$ is a shortest path from $(x,u)$ to $(y,v)$ if it has minimal length, called the distance from $(x,u)$ to $(y,v)$ and denoted by $d((x,u),(y,v))$. The path $P$ is a fastest path from $(x,u)$ to $(y,v)$ if it has minimal duration, called the latency from $(x,u)$ to $(y,v)$ and denoted by $\ell((x,u),(y,v))$. The path $P$ is a \sfp\ from $(x,u)$ to $(y,v)$ if it is a path of minimum length among those of minimal duration from $(x,u)$ to $(y,v)$.

For instance, in the case of Figure~\ref{fig:ex-intro}, the path $a,2,b,4,c,6,d,9,e$ is a fastest path from $(0,a)$ to $(32,e)$, but $a,1,b,4,\allowbreak c,6,d,9,e$ is not (it has duration $8$). The path $a,2,b,4,c,6,d,9,e$ has length $4$ and duration $7$, and no path from $(0,a)$ to $(32,e)$ with lower duration exists. It is not a shortest path since $a,9,c,18,d,23,e$ also is a path from $(0,a)$ to $(32,e)$ which has length $3$ and duration $14$. This last path is a shortest path, since no path with lower length exists, but not a fastest one. The distance from $(0,a)$ to $(32,e)$ therefore is $3$ and the latency is $7$. Among the fastest paths from $(0,a)$ to $(32,e)$, {\em i.e.} the paths of duration $7$, the shortest have length $4$. Therefore, $a,2,b,4,c,6,d,9,e$ is a \sfp\ between them, as well as $a,2,b,5,c,6,d,9,e$, for instance.

Finally, the betweenness of a node $v\in V$ at a time instant $t\in T$ measures how frequently $(t,v)$ is involved in \sfp s in $L$, see \cite{DBLP:journals/snam/LatapyVM18}:
$$
B(t,v)
= \sum_{u\in V, w\in V} \int_{i \in T, j\in T} \bfrac{u}{w} \diff i \diff j
$$
where $\bfrac{u}{w}$ is the fraction of all \sfp s from $u$ at time $i$ to $w$ at time $j$ that involve $v$ at time $t$ if there is a path from $(i,u)$ to $(j,w)$, $0$ otherwise.

In this original definition, the quantity $\bfrac{u}{w}$ is only loosely defined as a fraction of \sfp s; the function $\sigma$ itself, as well as the ratio between its values, are not explicitely defined. We will see in next section that this fraction involves uncountable sets of \sfp s that have finite volumes with a size and a dimension. We will also introduce the appropriate arithmetic operators needed to deal with them, and an algorithm to compute these volumes.

%%%%%%%%%%%%%%%%%%%%%%%%%%%%%%%%%%%%%%%% 
\section{Volumes of shortest paths}
\label{sec:volumes}

Let us consider a link stream $L=(T,V,E)$, and a sequence $I_1$, $I_2$, $\cdots$, $I_k$ of intervals of $T$.
Let us denote by $b_i$ and $e_i$ the bounds of interval $I_i$, with $b_i \le e_i$. If $e_i=b_i$ then $I_i$ is a singleton ($I_i = \{b_i\} = \{e_i\}$). The intervals may be closed ($I_i = [b_i,e_i]$), half-open ($I_i = ]b_i,e_i]$ or $I_i = [b_i,e_i[$), or open ($I_i = ]b_i,e_i[$).

We say that the sequence $I_1$, $I_2$, $\cdots$, $I_k$ is a {\bf sliding sequence} if for all $i$, there exists no element in $I_{i+1}$ strictly smaller than all elements of $I_i$ ($\nexists y \in I_{i+1}, \forall x \in I_i, y < x$), and no element of $I_i$ strictly larger than all elements of $I_{i+1}$ ($\nexists x \in I_i, \forall y \in I_{i+1}, x > y$).

In such a sequence, the intervals may overlap ($I_i\cap I_j \neq \emptyset$, $i\neq j$), may be included in each other ($I_i \subseteq I_j$, $i\neq j$), or may even be equal ($I_i= I_j$, $i\neq j$).

Given a sliding sequence $I_1$, $I_2$, $\cdots$, $I_k$, we denote by $v_0,I_1,v_1,\allowbreak I_2,v_2,\cdots I_k,v_k$ the set $S$ of all sequences $v_0,t_1,v_1,t_2,v_2,\cdots,t_k,v_k$ such that $v_i \in V$, $t_i\in I_i$ and $t_{i+1}\ge t_i$ for all $i$. We say that $S$ is a {\bf sliding set}. If the intervals are disjoint then $S = \{v_0\}\times I_1\times \{v_1\}\times\allowbreak I_2\times \{v_2\} \times \cdots \times I_k\times \{v_k\}$, but this is not true in general.

In the case of Figure~\ref{fig:ex-intro}, for instance, $[23,24], \intoc{25,28}, [27,29], \{30\}$ is a sliding sequence and $a,[23,24],b,\intoc{25,28},c,[27,29],d,\{30\},e$ is a sliding set. The elements of this set are the the paths $a,t_1,b,t_2,c,t_3,d,t_4,e$ with $23\le t_1\le 24$, $25< t_2\le 28$, $\max(27,t_2)\le t_3\le 29$, and $t_4=30$.

\medskip

More generally, all paths in any link stream are elements of sliding sets. In the case of Figure~\ref{fig:ex-intro}, for instance, all shortest paths from $(0,a)$ to $(14,e)$ go from $a$ to $b$ between times $1$ and $2$, from $b$ to $c$ between times $3$ and $5$, from $c$ to $d$ between $6$ and $7$, and finally from $d$ to $e$ between $9$ and $11$. Therefore, they are elements of 
$a,[1,2],\allowbreak b,[3,5],\allowbreak c,[6,7],\allowbreak d,[9,11],e$.

In addition, if we consider any two elements $(i,u)$ and $(j,v)$ of $T\times V$, then we have the following result.

\begin{proposition}
The set $\mySP((i,u),(j,v))$ of all shortest paths from $(i,u)$ and $(j,v)$ is the disjoint union of a finite number of sliding sets.
\end{proposition}
\begin{proof}
Let us consider all sliding sequences $I_1$, $I_2$, $\cdots$, $I_k$ with $k = d((i,u),(j,v))$ and $I_i$ is either an open interval $]t,t'[$ such that $t$ and $t'$ are two consecutive event times, or $I_i$ is a singleton $\{t\}$ such that $t$ is an event time. There is a finite number of such sequences, and they induce a finite number of sliding sets which are all disjoint.

Any path in $\mySP((i,u),(j,v))$ is in one of these sliding sets, and then all the elements of this sliding set are shortest paths from $(i,u)$ to $(j,v)$. Therefore $\mySP((i,u),(j,v))$ is the union of such sliding sets.
\end{proof}

\noindent
For instance, let us consider the following sliding sets:\\
$A = a,[1,2],b,[3,5],c,[6,7],d,[9,11],e$;
\allowbreak
$B = a,[8,9],c,\{11\},b,[12,14],d,\{16\},e$;\\
\allowbreak
$C = a,[15,16],b,\{19\},c,\{19\},d,[23,24],e$;
\allowbreak
$D = a,[23,24],b,[25,28],c,[27,29],d,[30,31],e$;\\
\
\allowbreak
$E = a,[1,2],b,[12,14],d,\{16\},e$;
\allowbreak
$F = a,[1,2],b,[12,14],d,\{23\},e$;\\
\allowbreak
$G = a,[1,2],b,[12,14],d,[23,24],e$;
\allowbreak
$H = a,[1,2],b,[12,14],d,[30,31],e$;\\
\allowbreak
\
$I = a,[8,9],c,[18,19],d,\{23\},e$;
\allowbreak
$J = a,[8,9],c,[18,19],d,[23,24],e$.\\
\allowbreak
$K = a,[8,9],c,[18,19],d,[30,31],e$;
\allowbreak
and
\allowbreak
$L = a,[8,9],c,[27,29],d,[30,31],e$.

Then, consider the link stream of Figure~\ref{fig:ex-intro}. There are simple cases where each set of shortest paths corresponds to a unique sliding set, like for instance
$\mySP((0,a),(14,e)) = A$,
$\mySP((4,a),(17,e)) = B$,
$\mySP((12,a),(26,e)) = C$,
$\mySP((20,a),(32,e)) = D$,
or
$\mySP((0,a),(18,e)) = E$.
In most cases, however, the set of shortest paths are disjoint unions (denoted by $\sqcup$) of several sliding sets, like for instance
$\mySP((0,a),(23,e)) = E \sqcup F \sqcup I$,
$\mySP((0,a),(26,e)) = E \sqcup G \sqcup J$,
or
$\mySP((0,a),(32,e)) = E \sqcup G \sqcup H \sqcup J \sqcup K \sqcup L$.

\begin{definition}[volumes]
The volume of a sliding set $S = v_0,I_1,v_1,\allowbreak I_2,v_2,\cdots I_k,v_k$, denoted by $|S|$, is defined by its {\bf size} and {\bf dimension} as follows:
\begin{itemize}
\item
If $I_i$ is a singleton for all $i$, then $S$ contains only one sequence. It has size $1$ and dimension $0$.
\item
Otherwise, let $I'_1, I'_2, \cdots, I'_l$ be the subsequence of $I_1, I_2, \cdots, I_k$ composed of all its intervals that are not singletons, and let $b'_i$ and $e'_i$, $b'_i < e'_i$, denote the bounds of $I'_i$, for all $i$. Then,
$\mysize(S) = \int_{t_1=b'_1}^{e'_1}\allowbreak \int_{t_2=max(t_1,b'_2)}^{e'_2}\dots\allowbreak\int_{t_l=max(t_{l-1},b'_l)}^{e'_l} 1 \diff t_l \dots \diff t_2 \diff t_1$ and
$\mydim(S) = l$.
\end{itemize}
In both cases, the volume of $S$, $|S|$, is defined as the pair $(\mysize(S),\mydim(S))$ giving its size and dimension.
\end{definition}

For instance, the sliding sets above have the following volumes: $|A|=(4,4)$, $|B| = (2,2)$, $|C| = (1,2)$, $|D| = (5.5,4)$, $|E| = (2,2)$, $|F| = (2,2)$, $|G| = (2,3)$, $|H| = (2,3)$, $|I| = (1,2)$, $|J| = (1,3)$, $|K| = (1,3)$, and $|L| = (2,3)$. The case of $D$ is different from the others, as it involves two non-trivially overlapping intervals, namely $[25,28]$ and $[27,29]$. Therefore, $D$ may be written as $D = a,[23,24],b,[25,27],c,[27,29],d,[30,31],e \cup a,[23,24],b,[27,28],c,[27,28],d,[30,31],e \cup a,[23,24],b,[27,28],c,[28,29],d,[30,31],e$. The volume of $D$ is then the sum of volumes of these three sliding sets. The first and last ones have volumes $(4,4)$ and $(1,4)$, respectively. The middle one has volume $(0.5,4)$, since it is the set of all sequences of the form $a,t_1,b,t_2,c,t_3,d,t_4$ with $t_1$ in $[23,24]$, both $t_2$ and $t_3$ in $[27,28]$, and $t_4$ in $[30,31]$, with the constraint that $t_2 \le t_3$.

More generally, we have the following definitions for volume operations.

\begin{definition}[addition, $\plusvol$]
\label{def:plusvol}
Given two disjoint sliding sets $S$ and $S'$ of volume $|S| = (s,d)$ and $|S'| = (s',d')$, the volume of their union $S \sqcup S'$ is the sum of their two volumes, which we denote by $|S| \plusvol |S'|$. In such a sum, volumes in lower dimensions are negligible, and the sizes of volumes with maximal dimension just add up, so we obtain $|S \sqcup S'| = |S| \plusvol |S'| = (s+s',d)$ if $d=d'$, $(s,d)$ if $d>d'$, and $(s',d')$ if $d'>d$. By extension, any disjoint union of a finite number of sliding sets $S_1$, $S_2$, $\cdots$, $S_k$ has dimension equal to the largest dimension of these sets, and size equal to the sum of the size of all these sets of maximal dimension; we denote its volume by $\bigplusvol_{i=1}^k |S_i| = |S_1| \plusvol |S_2| \plusvol \cdots \plusvol |S_k|$.
\end{definition}

\begin{definition}[product, $\timesvol$]
\label{def:timesvol}
Consider three nodes $u$, $v$ and $w$ in $V$, and two sets $S$ and $S'$ such that all elements of $S$ are of the form $u,t_1,v_1,t_2,\cdots,t_k,v$ and the ones of $S'$ are of the form $v,t'_1,v'_1,t'_2,\cdots,t'_l,w$, with $t_k \le t'_1$. We denote by $S \cdot S'$ the set of all sequences $u,t_1,v_1,t_2,\cdots,t_k,v,t'_1,v'_1,t'_2,\cdots,t'_l,w$ such that the sequence from $u$ to $v$ is in $S$ and the one from $v$ to $w$ is in $S'$. If $S$ and $S'$ are disjoint unions of a finite number of sliding sets with $|S| = (s,d)$ and $|S'| = (s',d')$, then $S \cdot S'$ also is the disjoint union of a finite number of sliding sets, and its volume is $|S \cdot S'| = |S| \timesvol |S'| = (s\cdot s', d+d')$.
\end{definition}

\begin{definition}[quotient and difference, $\divvol$, $\Tfrac{\,\,\,\,\,}{}$ and $\minusvol$]
\label{def:divvol}
Consider $S$ and $S'$ two disjoint unions of sliding sets with $|S| = (s,d)$ and $|S'| = (s',d')$, and such that $S' \subseteq S$. Then necessarily $d'\le d$ and the fraction of elements of $S$ that are also in $S'$, which we denote by $|S'| \divvol |S|$ or $\Tfrac{|S'|}{|S|}$, is equal to $0$ if $d>d'$, and to $s'/s$ if $d=d'$. In addition, the set $S \setminus S'$ is a disjoint union of sliding sets, and its volume is $(s,d) \minusvol (s',d')= (s,d) \plusvol (-s',d')$.
\end{definition}

{\bf These notations and operations make it easy to describe the set $\mySP((i,u),(j,v))$ and compute its volume, which is the goal of this section.}

In the non-trivial cases above, for instance,
$|\mySP((0,a),(23,e))| = |E \sqcup F \sqcup I| = |E| \plusvol |F| \plusvol |I| = (2,2) \plusvol (2,2) \plusvol (1,2) = (5,2)$,
$|\mySP((0,a),(26,e))| = |E \sqcup G \sqcup J| = |E| \plusvol |G| \plusvol |J| = (2,2) \plusvol (2,3) \plusvol (1,3) = (3,3)$,
and
$|\mySP((0,a),(32,e))| = |E \sqcup G \sqcup H \sqcup J \sqcup K \sqcup L| = (2,2) \plusvol (2,3) \plusvol (2,3) \plusvol (1,3) \plusvol (1,3) \plusvol (2,3) = (8,3)$.

\medskip

We will now prove two lemmas needed to compute the volume of shortest paths from a given temporal node $(i,u)$ in $T\times V$ to another one $(j,v)$ in $T\times V$.
Lemma~\ref{lem:tower} shows how to compute the volume of shortest paths between two consecutive event times. Lemma~\ref{lem:pred} shows how to decompose the set of shortest paths from a temporal node to another one into a disjoint union of smaller sets of shortest paths.
This will lead to Algorithm~\ref{alg:vsp}, that starts by computing the volume of shortest paths from $(i,u)$ to $(i,w)$ for any $w$. Then, in a temporal BFS-like manner, it uses volumes from $(i,u)$ to $(t,x)$ to compute volumes from $(i,u)$ to $(t',w)$, for increasing pairs of consecutive event times $t$ and $t'$. Indeed, as illustrated in Figure~\ref{fig:vsp-bfs}, the volumes at $t'$ can be derived from the ones at $t$. The temporal BFS also uses two queues, named $Q$ and $X$, to compute the distance that are also needed to compute volumes of shortest paths. It stops when it reaches time $j$.

In all the following, we consider two consecutive event times $t$ and $t'$. For all $x$ and $y$ in $]t,t'[$, the graphs $G_x$ and $G_y$ are identical. We denote by $G^+_t$ (or $G^-_{t'}$) this graph, and by $\sigma^+_t(u,v)$ and $d^+_t(u,v)$ (or $\sigma^-_{t'}(u,v)$ and $d^-_{t'}(u,v)$) the (finite) number of shortest paths and the distance from $u$ to $v$ in this graph.

\begin{lemma}
Given two nodes $x$ and $w$, the volume of the set of shortest paths from $x$ to $w$ that start and arrive during $]t,t'[$ is equal to
$$
\left(\ \sigma^+_t(x,w)\cdot \frac{(t'-t)^{d^+_t(x,w)}}{d^+_t(x,w)!},\ \ \ d^+_t(x,w)\ \right).
$$
\label{lem:tower}
\end{lemma}
\begin{proof}
First notice that if $v_0,t_1,v_1,t_2,v_2,\dots,t_k,v_k$ is a shortest path from $(t,x)$ to $(t',w)$ in $L$ with $t_i \in ]t,t'[$ for all $i$, then necessarily $v_0,v_1,v_2,\dots,v_k$ is a shortest path from $x$ to $w$ in $G_t^+$. Conversely, if $v_0,v_1,v_2,\dots,v_k$ is a shortest path from $x$ to $w$ in $G_t^+$ then each sequence $v_0,t_1,v_1,t_2,v_2,\dots,t_k,v_k$ with $t_i \in ]t,t'[$ and $t_i \le t_{i+1}$ for all $i$ is a shortest path from $(t,x)$ to $(t',w)$ in $L$.

Therefore, the set of shortest paths from $x$ to $w$ that start and arrive during $]t,t'[$ is the disjoint union of $v_0,]t,t'[,v_1,]t,t'[,\cdots,]t,t'[,v_k$ for all shortest path $v_0,v_1,v_2,\dots,v_k$ from $x$ to $w$ in $G_t^+$, where $v_0=x$, $v_k=w$, and $k=d^+_t(x,w)$.
It is easy to show by induction that the size of each such sliding set is $
\int_{t_1=t}^{t'}\allowbreak \int_{t_2=t_1}^{t'}\dots\allowbreak\int_{t_k=t_{k-1}}^{t'} 1 \diff t_k \dots \diff t_2 \diff t_1
=
\frac{(t'-t)^k}{k!}
$, and its dimension is $k$. The volume of $\mySP((t,x),(t',w))$ is the sum of the volumes of all these sliding sets, and there are $\sigma^+_t(x,w)$ such sliding sets, which completes the proof.
\end{proof}

\begin{lemma}
\label{lem:pred}
Given $(i,u)$ in $T\times V$ and $w$ in $V$, we define the two sets
$X = \{x\in V, d((i,u),(t,x))+d^+_t(x,w) = d((i,u),(t',w)) \}$
and
$Y = \{y\in N_{t'}(w), d((i,u),(t',y))+1 = d((i,u),(t',w)) \}$. 
Then, the volume of $\mySP((i,u),(t',w))$ is the sum of the two following volumes:
$$
\bigplusvol_{x\in X} \left( |\mySP((i,u),(t,x))| \timesvol
\left(\sigma^+_t(x,w)\cdot \frac{(t'-t)^{d^+_t(x,w)}}{d^+_t(x,w)!},d^+_t(x,w)\right)
 \right)
$$
and
$$
\bigplusvol_{y\in Y} |\mySP((i,u),(t',y))|.
$$
\label{lem:vsp-step}
\end{lemma}
\begin{proof}
Let us denote by $A$ the set $A = \sqcup_{x \in X} \mySP((i,u),(t,x)) \cdot \mySP^+((t,x),(t',w))$, where $\mySP^+((t,x),(t',w))$ is the set of shortest paths from $(t,x)$ to $(t',w)$ that start and arrive during $]t,t'[$. Let us denote by $B$ the set $B = \sqcup_{y\in Y} \mySP((i,u),(t',y)) \cdot \{(y,t',w)\}$, which means that $B$ is the set obtained when one concatenates any sequence in $\mySP((i,u),(t',y))$ with $y \in Y$ to the sequence $y,t',w$. By definition of $X$ and $Y$, elements of $A$ and $B$ are shortest paths from $(i,u)$ to $(t',w)$, and $A$ and $B$ are disjoint.

Let us now consider a shortest path $P = v_0,t_1,v_1,t_2,v_2,\cdots,t_k,v_k$ in $\mySP((i,u),(t',w))$, hence $v_0=u$ and $v_k=w$. We show that $P$ is in $A$ or $B$. Indeed, if $t_k = t'$ then $v_{k-1}$ is in $Y$, and $v_0,t_1,v_1,t_2,v_2,\cdots,v_{k-1}$ is in $\mySP((i,u),(t',v_{k-1}))$, which implies that $P$ is in $B$. If instead $t_k<t'$, let $l$ be the largest value such that $t_l \le t$. Then, all $t_j$ with $l<j\le k$ are in $]t,t'[$ and $v_l$ necessarily is in $X$. Therefore, $v_l,t_{l+1},\cdots,t_k,v_k$ necessarily is a shortest path from $(t,v_l)$ to $(t',w)$ that starts and arrives in $]t,t'[$. In addition, $v_0,t_1,v_1,\cdots,t_l,v_l$ is a shortest path from $(i,u)$ to $(t,v_l)$, with $v_l \in X$. Therefore, $P$ is in $A$.

Finally, $\mySP((i,u),(t',w))$ is exactly $A \sqcup B$, which, together with Lemma~\ref{lem:tower} and Definitions~\ref{def:plusvol} and~\ref{def:timesvol} on volume operations, proves the claim.
\end{proof}

\begin{figure}[!h]
\centering
%\includegraphics[angle=180,width=0.45\textwidth]{pred2s.jpg}
%\hfill
%\includegraphics[angle=180,width=0.45\textwidth]{pred1s.jpg}

%\includegraphics[width=0.45\textwidth]{pred2bis_s.jpg}
%\hfill
%\includegraphics[width=0.45\textwidth]{pred1bis_s.jpg}

\ \hfill
\includegraphics[width=0.3\textwidth]{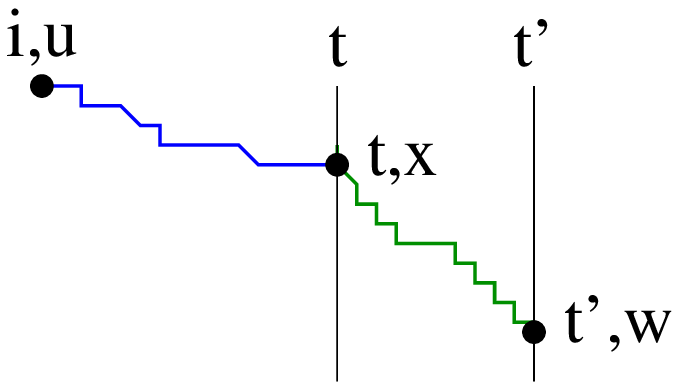}
\hfill
\includegraphics[width=0.3\textwidth]{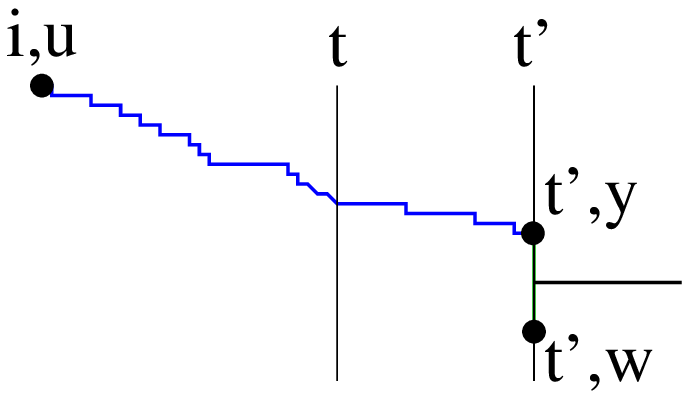}
\hfill\ 

\caption{
If $t$ and $t'$ are two consecutive event times, then a shortest path from $(i,u)$ to $(t',w)$ is the concatenation of either (left) a blue path from $(i,u)$ to a given $(t,x)$ and a green path from this $(t,x)$ to $(t',w)$ in $G^+_t$; or (right) a blue path from $(i,u)$ to a specific $(t',y)$ and then a jump from $y$ to $w$ at time $t'$ using $(t',yw) \in E_{t'}$.
}
\label{fig:vsp-bfs}
\end{figure}

\begin{algorithm}[!h]
\Fn{$\mbox{VSP}$}{
   \KwIn{a link stream $L=(T,V,E)$, $(i,u)\in T \times V$, and $(j,v)\in T \times V$}
   \KwOut{volume of shortest paths from $(i,u)$ to $(j,v)$}
   $\myDist \gets$ Dictionary initialized to $\infty$ for any key\\
   $\myvol \gets$ Dictionary intialized to $(0,0)$ for any key\\
   \For{each $w$ reachable from $u$ in $G_i$\label{alg:vsp:init}}{
    $\myDist[(i,w)] \gets d_i(u,w)$ and $\myvol[(i,w)] \gets (\sigma_i(u,w),0)$
    }

   \For{each $t$, $t'$ consecutive times in $\{i,j\} \cup (\eventtimes \cap [i,j])$ in increasing order \label{alg:vsp:mainloop}}{

    $Q \gets $ empty queue\label{alg:vsp:dist:begin}\\
    set all nodes as unmarked\\
    $X \gets $ list of all $(w,\myDist[(t,w)])$ in increasing order of $\myDist[(t,w)]$\\
    \While{$Q$ or $X$ is not empty}{
     $(w,d) \gets $ get and remove the first element of $Q$ or $X$ with minimal $d$\label{alg:vsp:dist:min}\\
     \lIf{$w$ is unmarked}{$\myDist[(t',w)] \gets d$ and mark $w$\label{alg:vsp:dist:found}}
     \lFor{all unmarked node $y$ in $N_{t'}(w)$}{add $(y,d+1)$ to $Q$\label{alg:vsp:dist:addQ}\label{alg:vsp:dist:end}}
     }

    \For{all marked node $w$ in increasing order of $\myDist[(t',w)]$\label{alg:vsp:vol:begin}}{
     \For{all marked node $x$ such that $\myDist[(t,x)]+d^+_t(x,w) = \myDist[(t',w)]$}{
      $\myvol[(t',w)] \gets \myvol[(t',w)] \plusvol \myvol[(t,x)] \timesvol \left(\sigma^+_t(x,w) \cdot \frac{(t'-t)^{d^+_t(x,w)}}{d^+_t(x,w)!},d^+_t(x,w)\right)$
      }
     \For{all marked node $y$ in $N_{t'}(w)$ such that $\myDist[(t',w)] = \myDist[(t',y)]+1$}{
      $\myvol[(t',w)] \gets \myvol[(t',w)] \plusvol \myvol[(t',y)]$\label{alg:vsp:vol:end}
      }

    }
  }

  \Return{$\myvol[(j,v)]$}
}
\caption{Volume of shortest paths between two temporal nodes.}
\label{alg:vsp}
\end{algorithm}

\begin{theorem}
Given two temporal nodes $(i,u)$ and $(j,v)$ in $T\times V$, Algorithm~\ref{alg:vsp} computes the volume of shortest paths from $(i,u)$ to $(j,v)$.
\end{theorem}
\begin{proof}
Let us consider any time $t<j$ in $\{i,j\} \cup (\eventtimes \cap [i,j])$ and let $t'$ be the next time in this set. We show below that, if $\myDist[(t,w)] = d((i,u),(t,w))$ and $\myvol[(t,w)] = |\mySP((i,u),(t,w))|$ for all $w$ when one enters the main loop at line~\ref{alg:vsp:mainloop}, then at the end of the loop we have $\myDist[(t',w)] = d((i,u),(t',w))$ and $\myvol[(t',w)] = |\mySP((i,u),(t',w))|$. This is sufficient to prove that the algorithm returns $|\mySP((i,u),(j,v))|$, since the loop at line~\ref{alg:vsp:init} initializes $\myDist$ and $\myvol$ correctly.

Lines~\ref{alg:vsp:dist:begin} to~\ref{alg:vsp:dist:end} deal with the computation of $d((i,u),(t',w))$ from the distances at time $t$, for all $w$. It is similar to a BFS on the graph $G_{t'}$, except that distances at $t'$ are bounded by the ones at $t$: $d((i,u),(t',w)) \le d((i,u),(t,w))$. The loop therefore uses two queues: a list $X$ of nodes in increasing distance at time $t$, and a queue $Q$ for the exploration of $G_{t'}$. At each round, we consider a node $w$ with minimal distance in these queues: Line~\ref{alg:vsp:dist:min} takes the first element of $X$ or $Q$, depending on which has the minimal second field $d$. This is its actual distance $d((i,u),(t',w))$ (line~\ref{alg:vsp:dist:found}). Then we add its neighbors to $Q$, together with the information that their distance from $(i,u)$ cannot be larger than $d((i,u),(t',w)) + 1$ (line~\ref{alg:vsp:dist:addQ}). The loop ends when both $X$ and $Q$ are empty, {\em i.e.} the distances to all reachable nodes are found.

Then, Lines~\ref{alg:vsp:vol:begin} to~\ref{alg:vsp:vol:end} deal with the computation of $|\mySP((i,u),(t',w))|$ from the volumes at time $t$, for all reachable $w$. They are a straightforward application of Lemma~\ref{lem:pred}.
\end{proof}

%\clearpage
%%%%%%%%%%%%%%%%%%%%%%%%%%%%%%%%%%%%%
\section{Latency pairs}
\label{sec:latencies}

Let us consider a link stream $L=(T,V,E)$, and two nodes $u$ and $w$ in $V$.
The previous section shows how to compute the volume of shortest paths from $u$ to $w$ between two given time instants $i$ and $j$. However, betweenness computations rely on volumes of shortest {\em fastest} paths from $u$ to $w$. These paths are the shortest paths from $(s,u)$ to $(a,w)$ if the latency from $(s,u)$ to $(a,w)$ is equal to $a-s$. We then say that  $(s,a)$ in $T\times T$ is a {\bf latency pair} from $u$ to $w$ (in $L$). This section is devoted to the computation of such latency pairs.

%A {\bf latency pair} from $u$ to $w$ (in $L$) is an ordered pair $(s,a)$ in $T\times T$ such that the latency from $(s,u)$ to $(a,w)$ is equal to $a-s$.

In the case of Figure~\ref{fig:ex-intro}, for instance, $(2,9)$ is a latency pair from $a$ to $e$, because the fastest paths from $(2,a)$ to $(9,e)$ start at $2$ and end at $9$. Similarly, $(9,16)$, $(16,23)$ and $(24,30)$ are the other latency pairs from $a$ to $e$. Instead, $(3,8)$ is not a latency pair from $a$ to $e$ since there is no path from $(3,a)$ to $(8,e)$, and $(1,9)$ is not a latency pair from $a$ to $e$ either because the fastest paths from $(1,a)$ to $(9,e)$ start at time $2$.

For any $t$ in $T$, the pair $(t,t)$ is a latency pair from $u$ to $w$ exactly if there is an instantaneous path between $(t,u)$ and $(t,w)$, {\em i.e.} there is a path between $u$ and $w$ in $G_t$. The latency between $(t,u)$ and $(t,w)$ is then equal to $0$, and we call $(t,t)$ an {\bf instantaneous latency pair}. In the case of Figure~\ref{fig:ex-intro}, such latency pairs occur from $b$ to $d$ at all times from $12$ to $14$, at time $19$, and at all times from $27$ to $28$.

Notice that there may exist an infinite amount of instantaneous latency pairs from a node to another one, like in this last example, but there is only a finite number of non-instantaneous latency pairs. Indeed, if $(s,a)$ is a latency pair with $a-s\ne 0$, then $s$ and $a$ necessarily are event times, and as said in Section~\ref{sec:preliminaries} all link streams considered here have a finite number of event times.

Notice also that if $(s,a)$ is a latency pair from $u$ to $w$, then there cannot be any latency pair $(s',a')$ from $u$ to $w$ with $[s',a'] \subsetneq [s,a]$. Indeed, this would imply that the latency from $(s,u)$ to $(a,w)$ is equal to $s'-a' < s-a$, which contradicts the fact that $(s,a)$ is a latency pair. This also implies that, if $(s,a)$ is a latency pair with $s\ne a$, then necessarily $s$ and $a$ are event times: otherwise, there is a pair $(s',a')$ such that $[s',a'] \subsetneq [s,a]$, with $G_{s'} = G_s$ and $G_{a'} = G_a$, which would imply that $(s',a')$ also is a latency pair, which contradicts our previous remark.

As a consequence, latency pairs are componentwise ordered: if $(s,a)$ and $(s',a')$ are two distinct latency pairs, then $\intcc{s',a'} \not\subseteq \intcc{s,a}$ and $\intcc{s,a} \not\subseteq \intcc{s',a'}$. Therefore, either $s<s'$ and $a<a'$, or $s'<s$ and $a'<a$.

In this section, we compute the {\bf latency list} from $u$ to $w$, defined as the (finite) componentwise ordered list of all latency pairs $(s,a)$ such that $s$ and $a$ are event times.
For instance, in the case of Figure~\ref{fig:ex-intro}, the latency list from $a$ to $e$ is $(2,9),(9,16),(16,23),(24,30)$, and the latency list from $b$ to $d$ is $(5,6),(12,12),(14,14),(19,19),(27,27),(28,28)$.

Our algorithm considers all event times in increasing order. It maintains the latency lists from a given node to all others before the current event time. It then updates these latency lists for the current time by computing the connected components of the graph at this time. For each of these components, it considers the latest starting time from which a node in this component can be reached, which is given by the previously computed latency lists. This time is the beginning of latency pairs for its nodes, that ends at current time, and so the algorithms updates the lists accordingly.

\DontPrintSemicolon
\begin{algorithm}
\Fn{Latency-lists}{
\KwIn{a link stream $L=(T,V,E)$ and $u\in V$}
\KwOut{for each $w \in V$, the latency list from $u$ to $w$}
create $\myLL \gets \mbox{empty dictionary}$ \label{line:latpairsbegin} and $\myLL[w] \gets \mbox{empty list for all }w$\\
\For{$t$ in $\eventtimes$\label{line:latencyfort}}{
 append $(t,t)$ to $\myLL[u]$\\
 \For{each connected component $C$ of $G_t$}{
  $s \gets \myNone$ and $X \gets \emptyset$\\
  \For{$w \in C$ with non-empty $\myLL[w]$ \label{line:propagatecreateX}}{
   $(s',a') \gets \mbox{last element of } \myLL[w]$\\
   \lIf{$s=\myNone$ or $s'>s$}{$s \gets s'$ and $X \gets \{w\}$}
   \lElseIf{$s'=s$}{add $w$ to $X$}
   }
   \If{$X$ is non-empty}{
   \lFor{$w \in C\setminus X$}{append $(s,t)$ to $\myLL[w]$}\label{line:latpairsend}
   }
  }
 }

 \Return{\myLL}
 }
\caption{Computation of all latency lists from a given node.}
\label{algo:latency-lists}
\end{algorithm}

\begin{theorem}
Given a link stream $L=(T,V,E)$ and a node $u\in V$, Algorithm~\ref{algo:latency-lists} computes the ordered latency lists from $u$ to any node $w \in V$.
\end{theorem}
\begin{proof}
We claim that, at the end of each iteration of the main loop, for all $w$ in $V$, $\myLL[w]$ is the list of all latency pairs $(s,a)$ from $u$ to $w$ such that $s$ and $a$ are event times with $a \le t$.

Assume this is true for all iterations before a given event time $t$. When it reaches this event time, the loop starts by adding $(t,t)$ to $\myLL[u]$, which makes the claim true for $w=u$. Consider any connected component $C$ of $G_t$; the nodes $w\in C$, with non-empty $\myLL[w]$ are the nodes reachable from $u$ with an arrival time before $t$ or at $t$.
Then, the value of $s'$ computed by the loop at Line~\ref{line:propagatecreateX} is the latest starting time such that one of these nodes is reachable from $(s',u)$ before $t$ or at $t$, and $X$ is the set of these reachable nodes.

Therefore, if $X$ is non-empty, there exists a path from $(s',u)$ to $(t,w)$ for any $w\in C\backslash X$: for any $x \in X$, the path from $(s',u)$ to $(t,x)$ and then from $(t,x)$ to $(t,w)$ (which exists since $x$ and $w$ are in the same connected component $C$ of $G_t$) is such a path. As a consequence, $(s',t)$ is a latency pair for any $w\in C \backslash X$.
Notice that $(s',t)$ is not a latency pair for any node $x\in X$, $x\ne u$, since they all have a latency pair $(s', t_x)$ with $t_x<t$.

Finally, if the claim is true for all event times lower than $t$, it is true for $t$ too. It is true for the first iteration, {\em i.e.} when $t$ is the first event time:
it sets $\myLL[w]$ to $\{(t,t)\}$ for all node $w$ in the same connected component of $G_t$ as $u$, which is the correct value. Therefore, for all $w$ in $V$, the returned value of $\myLL[w]$ is the list of latency pairs $(s,a)$ from $u$ to $w$ such that $s$ and $a$ are event times, and it is ordered by construction.

\end{proof}

\section{Contribution of a node pair}
\label{sec:contribution}

In all this section, we consider a link stream $L=(T,V,E)$ and two nodes $u$ and $w$ in $V$. In addition, we consider a temporal node $(t,v)$ in $T \times V$.

For any $i$ and $j$ in $T$, we denote by $C^{ij}_{tv}(u,w)$ the fraction $\frac{\sigma((i,u),(j,w),(t,v))}{ \sigma((i,u),(j,w))}$ of \sfp s from $(i,u)$ to $(j,w)$ that involve $(t,v)$, and we call it the {\bf contribution of $(i,j)$}. If there is no path from $(i,u)$ to $(j,w)$, we consider that $C^{ij}_{tv}(u,w)=0$. By extension, we call $\int_{i,j \in T} C^{ij}_{tv}(u,w) \diff i \diff j$ the {\bf contribution of $(u,w)$} to the betweenness of $(t,v)$, and we denote it by $C_{tv}(u,w)$. The goal of this section is to compute $C_{tv}(u,w)$.

First notice that the contribution of $(i,j)$ is derived from volumes of paths as follows. Given $x$, $y$ and $z$ in $T\times V$, we denote by $\mySFP(x,y)$ the set of all \sfp s from $x$ to $y$, and by $\mySFP(x,y,z)$ the set of these paths that involve $z$. Then, we define $\sigma(x,y)$ and $\sigma(x,y,z)$ as the volumes of $\mySFP(x,y)$ and $\mySFP(x,y,z)$, respectively. It follows that $C^{ij}_{tv}(u,w)$ is equal to $\sigma((i,u),(j,w),(t,v)) \divvol \sigma((i,u),(j,w))$ if there is a path from $(i,u)$ to $(j,w)$. Otherwise, $C^{ij}_{tv}(u,w)$ is $0$.

This gives a rigorous ground to the definition of $C^{ij}_{tv}(u,w)$, which, as discussed at the end of Section~\ref{sec:preliminaries}, was loosely defined as the fraction $\frac{\sigma((i,u),(j,w),(t,v))}{ \sigma((i,u),(j,w))}$ of \sfp s from $(i,u)$ to $(j,w)$ that involve $(t,v)$; it is indeed equal to the ratio between the two volumes $\sigma((i,u),(j,w),(t,v))$ and $\sigma((i,u),(j,w))$ now defined, with volume ratio operation from Definition~\ref{def:divvol}: $C^{ij}_{tv}(u,w) = \Tfrac{\sigma((i,u),(j,w),(t,v))}{ \sigma((i,u),(j,w))}$.

Consider for instance the case of Figure~\ref{fig:ex-intro} with $u=a$ and $w=e$, and let us consider $i=0$ and $j=18$. Then, the \sfp s from $(i,u) = (0,a)$ to $(j,w) = (18,e)$ are the elements of the set $\mySFP((0,a),(18,e)) = X \sqcup Y$ where $X$ and $Y$ are the sliding sets $a,\{2\},b,[3,5],c,[6,7],d,\{9\},e$ and $a,\{9\},c,\{11\},b,[12,14],d,\{16\},e$, respectively. If $(t,v)=(7.5,c)$ or $(t,v)=(10,b)$, for instance, then none of these paths involve $(t,v)$ and so we obtain a $0$ contribution. If $(t,v) = (4.5,c)$ or $(t,v)=(8,d)$, for instance, then all paths in $X$ involve $(t,v)$ and no path in $Y$ does, leading to $C^{ij}_{tv}(u,w) = \sigma((0,a),(18,e),(t,v)) \divvol \sigma((0,a),(18,e)) = |X| \divvol (|X|\plusvol |Y|) = (2,2) \divvol ((2,2)\plusvol (2,1)) = (2,2) \divvol (2,2) = 1$. If $(t,v) = (10,c)$ or $(t,v)=(14,d)$, then $C^{ij}_{tv}(u,w) = |Y| \divvol (|X|\plusvol |Y|) = 0$.

\medskip

Before presenting the algorithm computing these path volumes and associated contributions, we characterize more precisely which pairs $(i,j)$ have non-zero contribution.

\begin{lemma}
\label{lem:uniquelatsfp}
There is at most one latency pair from $u$ to $w$ with non-zero contribution.
\end{lemma}
\begin{proof}
Consider two distinct latency pairs $(s,a)$ and $(s',a')$; we can assume $s<s'$ and $a<a'$, since, as explained in previous section, $[s',a']\subseteq [s,a]$ is impossible. Suppose both latency pairs have non-zero contribution: there are \sfp s from $(s,u)$ to $(a,w)$ that involve $(t,v)$ and from $(s',u)$ to $(a',w)$ that also involve $(t,v)$.
Therefore, there is a path from $(s',u)$ to $(t,v)$ and a path from $(t,v)$ to $(a,w)$, and so a path from $(s',u)$ to $(a,w)$. It has duration $a-s'$ which is strictly lower than both $a-s$ and $a'-s'$, thus contradicting both that $(s,a)$ and $(s',a')$ are latency pairs.
\end{proof}

If all latency pairs from $u$ to $w$ have contribution $0$, then the contribution of $(u,w)$ itself is $0$. Otherwise, let us denote by $(s,a)$ the unique latency pair with non-zero contribution.

We now introduce {\bf two specific times, $S$ and $A$}, that we will use to find all time instants with non-zero contribution.
We define $\intoo{S,A}$ as the largest interval containing $\intoo{s,a}$ such that: for all other latency pair $(s',a')$ in this interval, either $a'-s' > a-s$, or $a'-s' = a-s$ and $d((s',u),(a',w)) \ge d((s,u),(a,w))$; and the number of instantaneous paths from $(S,u)$ to $(A,w)$ of length $d((s,u),(a,w))$ is finite.
We illustrate this definition in Figure~\ref{fig:contrib}.

\begin{figure}[!h]
\centering
\includegraphics[width=.95\linewidth]{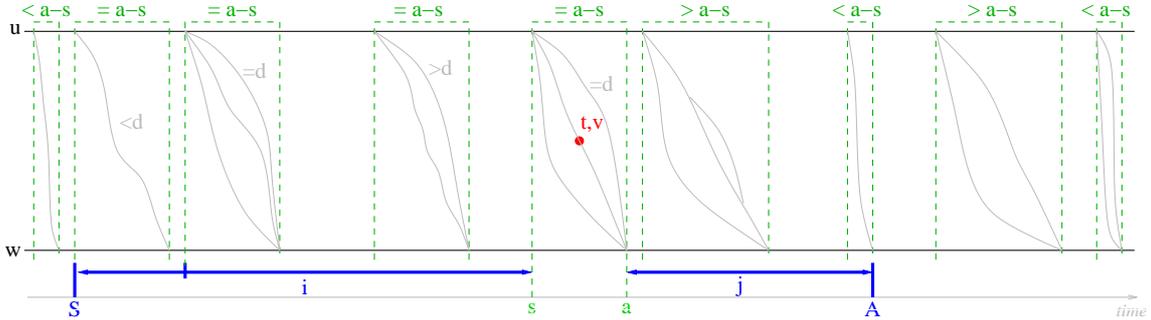}
\caption{
An abstract example of link stream $L = (T,V,E)$ in which we consider a specific $(t,v)$ in $T\times V$ (in red), two nodes $u$ and $w$ in $V$ (in black, horizontal lines), as well as the latency pair $(s,a)$ containing $t$ such that shortest (necessarily fastest) paths from $(s,u)$ to $(a,w)$ have length $d$ and some of them involve $(t,v)$. We display all latency pairs from $u$ to $w$ with two green vertical lines topped by a dotted horizontal line indicating the corresponding latencies ($=a-s$, $<a-s$ or $>a-s$). In addition, we also indicate the length ($=d$, $<d$ or $>d)$ of corresponding shortest paths within each latency pair, when this is useful (in grey). We indicate in blue the two specific times $S$ and $A$ defined above, as well as the time periods for $i$ and $j$ such that the contribution of $(i,j)$ may be non-zero (Lemma~\ref{lem:minsmina}).
}
\label{fig:contrib}
\end{figure}

We then have the following result.

\begin{lemma}
\label{lem:minsmina}
All pairs $(i,j)$ in $T \times T$ that have non-zero contribution are in $\intcc{S,s}\times \intcc{a,A}$.
\end{lemma}

\begin{proof}
If a given pair $(i,j)$ has non-zero contribution, then there is a latency pair $(s',a')$ with $s'\ge i$ and $a'\le j$ that has non-zero contribution. Remind that $(s,a)$ is itself such a latency pair. From Lemma~\ref{lem:uniquelatsfp}, we then have $(s',a') = (s,a)$, and so $i \le s$ and $j \ge a$.

If $i<S$ and $j \ge a$, or if $i \le s$ and $j > A$, then by definition of $S$ and $A$ we are in one of the following situations.

There exists a latency pair $(s',a')$ in $\intcc{i,j}$ such that: either $a'-s' < a-s$, or $a'-s' = a-s$ and $d((s',u),(a',w)) < d((s,u), (a,w))$. Then, \sfp s from $(s,u)$ to $(a,w)$ are not \sfp s from $(i,u)$ to $(j,w)$. All \sfp s from $(i,u)$ to $(j,w)$ are from $(s',u)$ to $(a',w)$ where $(s',a')$ is a latency pair as described above. Suppose such a \sfp\ involves $(t,v)$. Then there are paths from $(s',u)$ to $(t,v)$ and from $(t,v)$ to $(a',w)$. As a consequence, $s' \not\in \intcc{s,t}$, otherwise $(s,a)$ would not be a latency pair. Likewise, $a' \not\in \intcc{t,a}$. Therefore, $s' < s$ and $a'>a$, but this contradicts the fact that $a'-s' \le s-a$. This means that \sfp s from $(s',u)$ to $(a',w)$ cannot involve $(t,v)$, and so the contribution of $(i,j)$ is $0$.

Or there is an infinite number of instantaneous paths from $(i,u)$ to $(j,w)$ with length $d((s,u),(a,w))$. Only the $\sigma_t(u,w)$ ones starting and arriving at time $t$ involve $(t,v)$. There is a finite number of such paths, as they are paths in the graph $G_t$. Therefore, the contribution of $(i,j)$ is zero.

In conclusion, $i \le s$, $j \ge a$, $i$ cannot be smaller than $S$, and $j$ cannot be larger than $A$, which proves the claim.
\end{proof}

This lemma says that all pairs $(i,j)$ with non-zero contribution are in $\intcc{S,s}\times \intcc{a,A}$. Notice however that some pairs $(i,j)$ in $\intcc{S,s} \times \intcc{a,A}$ may have a contribution equal to 0. This happens whenever the volume of \sfp s from $(s,u)$ to $(a,w)$ has a lower dimension than the one from $(i,u)$ to $(j,w)$.

We now define specific latency pairs that play a special role, as any shortest fastest path from $(i,u)$ to $(j,w)$ must start and arrive within one of these pairs. To do this, we introduce an {\bf ordered list $\latpairset$ of latency pairs} centered on $(s,a)$, which means that latency pairs preceding $(s,a)$ have negative indexes in the list and the others have positive indexes. It is the list $\latpairset = (s_{-l},a_{-l}), (s_{-l+1}, a_{-l+1})), \dots, (s_0=s, a_0=a) ,\dots, (s_r,a_r)$ such that, for all $k$, $\intcc{s_k,a_k} \subseteq \intcc{S,A}$, $a_k-s_k = a-s$, and $d((s_k,u),(a_k,w)) = d((s,u),(a,w))$. We also define $s_{-l-1}=S$ and $a_{r+1} = A$.
Notice that $s_{-l-1} = s_l$ or $a_r = a_{r+1}$ are not forbidden; this happens for instance when $s_{-l} = \alpha$ or $a_r=\omega$. We show now that the latency pairs in $\latpairset$ give precisely the shortest fastest paths from $u$ to $w$.

\begin{lemma}
\label{lem:partition}
For any pair $(i,j)$ in $\intoc{S,s} \times \intco{a,A}$, the set $\mySFP((i,u),(j,w))$ is the disjoint union of all sets $\mySFP((s_k,u),(a_k,w))$ such that $(s_k,a_k)$ in $\latpairset$ and $\intcc{s_k,a_k} \subseteq [i,j]$.
\end{lemma}
\begin{proof}
We first show that for any $k$ such that $\intcc{s_k,a_k} \subseteq [i,j]$, $\mySFP((s_k,u),(a_k,w)) \subseteq \mySFP((i,u),(j,w))$. Let us consider a path in $\mySFP((s_k,u),(a_k,w))$. Since $\intcc{s_k,a_k} \subseteq [i,j]$, it is a path from $(i,u)$ to $(j,w)$. It has duration $s_k-a_k = s-a$ because $(s_k,a_k)$ is in $\latpairset$. Moreover, since $i>S$ and $j<A$, there exists no latency pair $(s',a')$ such that $[s',a'] \subseteq [i,j]$ and $a'-s'< a-s$. Therefore, it is a fastest path from $(i,u)$ to $(j,w)$. Similarly, because $(s_k,a_k)$ is in $\latpairset$, this path has length $d((s,u),(a,w))$ and therefore it is a \sfp\ from $(i,u)$ to $(j,w)$.

Now consider any \sfp\ from $(i,u)$ to $(j,w)$, and let us denote by $s'$ and $a'$ its starting and arrival times. Since it is a fastest path, $(s',a')$ is a latency pair, and obviously $[s',a'] \subseteq [i,j]$. In addition, $s'-a' = s-a$: if it was larger then the paths from $(s',u)$ to $(a',w)$ would not be fastest paths from $(i,u)$ to $(j,w)$; and if it was smaller, then the paths from $(s,u)$ to $(a,w)$ would not be fastest paths from $(i,u)$ to $(j,w)$. Similarly, $d((s',u),(a',w)) = d((s,u),(a,w))$: if it was larger then the paths from $(s',u)$ to $(a',w)$ would not be shortest paths from $(i,u)$ to $(j,w)$; and if it was smaller, then the paths from $(s,u)$ to $(a,w)$ would not be shortest paths from $(i,u)$ to $(j,w)$. Therefore, $(s',a')$ is in $\latpairset$, leading to the fact that $\mySFP((i,u),(j,w))$ is included in the union of all sets $\mySFP((s_k,u),(a_k,w))$ such that $(s_k,a_k)$ in $\latpairset$ and $\intcc{s_k,a_k} \subseteq [i,j]$.

Finally, notice that the sets $\mySFP((s_k,u),(a_k,w))$ are disjoint for different values of $k$, since all the paths they contain start at $s_k$ and arrive at $a_k$. We therefore obtain the claim.
\end{proof}

\begin{algorithm}
\Fn{PrevList}{
\KwIn{a link stream $L = (T,V,E)$, $u \in V$, $w \in V$, $(s,a)$  a latency pair from $u$ to $w$, and the ordered latency list $\myLL$ from $u$ to $w$}
\KwOut{the  list $((s_{-1}, f_{-1}), (s_{-2}, f_{-2}), \ldots, (s_{-l-1} = S, f_{-l-1})$ with $s_k$ defined by \latpairset{} and with $f_k = \sigma((s_{k+1},u), (a,w)) \minusvol \sigma((s,u),(a,w))$}
  init $\myR$ to empty list and $\myvol$ to $(0,0)$\\
  \If{$s = a$ and $d^-_{s}(u, w) = d((s,u),(a,w))$ \label{alg:prevlist:special1}}{
    \Return{$\myR$} \label{alg:prevlist:special2}
    }
  \ForEach{$(s',a')$ with $s'<s$ in $\myLL$ backwards \label{alg:prevlist:for}}{
    \If{$a'-s'<a-s$}{
      append $(s',\myvol)$ to $\myR$ \label{alg:prevlist:append:lat} and \Return{$\myR$} \label{alg:prevlist:returnlat}
    }
    \If{$a' - s'= a-s$}{
      \If{$d((s',u),(a',w)) < d((s,u),(a,w))$}{
        append $(s', vol)$ to $\myR$ \label{alg:prevlist:append:dist} and \Return{$\myR$} \label{alg:prevlist:return:dist}\\
      }
      \If{$d((s',u),(a',w)) = d((s,u),(a,w))$}{
          append $(s',\myvol)$ to $\myR$ \label{alg:prevlist:append:inf}\\
          \If{$s' = a'$ and $d^-_{s'}(u, w) = d((s,u),(a,w))$              \label{alg:prevlist:lineeps:2}}{
              \Return{$\myR$} \label{alg:prevlist:returninf}
          }
          $\myvol \gets \myvol \plusvol \mbox{VSP}(L, (s',u), (a',w))$ \label{alg:prevlist:updatevol}\\
        }
    }
  }
  append $(\alpha,\myvol)$ to $\myR$ \label{alg:prevlist:append:alpha} and \Return{$\myR$} \label{alg:prevlist:returnalpha}

}
\caption{Compute the list of starting times for latency pairs in \latpairset, as well as associated volumes of sets of shortest fastest paths.}
\label{alg:prevlist}
\end{algorithm}

Thanks to these results, we obtain an expression giving the contribution of a node pair as a discrete sum.

\begin{lemma}
  \label{lem:contrib:sum}
The contribution of $(u,w)$ to the betweenness of $(t,v)$, {\em i.e.} the fraction of \sfp s from $u$ to $w$ that involve $(t,v)$, namely
$C_{tv}(u,w) =
\int_{i,j \in T} \Tfrac{\sigma((i,u),(j,w),(t,v))}{ \sigma((i,u),(j,w))} \diff i \diff j
$,
can be written as a discrete sum:
\begin{equation}
  C_{tv}(u,w) =
  \sum_{k=-l-1}^{-1} \sum_{k'=0}^r (s_{k+1} - s_k ) (a_{k'+1} - a_{k'} )
  \Dfrac{\sigma((s,u),(a,w),(t,v))}{ \plusvol_{h=k+1}^{k'} \sigma((s_h,u),(a_h,w))}.
\label{eq:sumcontrib}
\end{equation}
\end{lemma}
\begin{proof}
According to Lemma~\ref{lem:minsmina}, the contribution of time instants $(i,j)$ is equal to zero whenever $(i,j)\not\in \intcc{S,s} \times \intcc{A,a}$.
For $(i,j) \in \intcc{S,s} \times \intcc{A,a}$,
all shortest fastest paths from $(i,u)$ to $(j,w)$ involving $(t,v)$ start at time $s$ and arrive at time $a$  and
the contribution of $(i,j)$ is therefore equal to
$\sigma((s,u),(a,w),(t,v)) \divvol{} \sigma((i,u),(j,w))$.
Therefore, $C_{tv}(u,w)
= \int_{\intcc{S,s}\times \intcc{a,A}} \Tfrac{\sigma((s,u),(a,w),(t,v))}{ \sigma((i,u),(j,w))} \diff i \diff j
= \int_{\intoc{S,s}\times \intco{a,A}} \Tfrac{\sigma((s,u),(a,w),(t,v))}{ \sigma((i,u),(j,w))} \diff i \diff j$.

According to Lemma~\ref{lem:partition},
for any $k<0$,
any $k' \ge 0$, any $i\in \intoc{s_k, s_{k+1}}$, and any  $j\in \intco{a_{k'}, a_{k'+1}}$,
the value of $\sigma((i,u),(j,w))$ is constant and it is equal to
$\plusvol_{h=k+1}^{k'} \sigma((s_h,u),(a_h,w))$.
  
Therefore,
\begin{equation*}
\begin{split}
\int_{s_k}^{s_{k+1}} \int_{a_{k'}}^{a_{k'+1}} \Dfrac{\sigma((i,u),(j,w),(t,v))}{ \sigma((i,u),(j,w))} \diff i \diff j &
=
\int_{s_k}^{s_{k+1}} \int_{a_{k'}}^{a_{k'+1}} \Dfrac{\sigma((s,u),(a,w),(t,v))}{ \sigma((i,u),(j,w))} \diff i \diff j\\
&
=
\int_{s_k}^{s_{k+1}} \int_{a_{k'}}^{a_{k'+1}} \Dfrac{\sigma((s,u),(a,w),(t,v))}{ \plusvol_{h=k+1}^{k'} \sigma((s_h,u),(a_h,w)) } \diff i \diff j\\
&
=
(s_{k+1} - s_k ) (a_{k'+1} - a_{k'} ) \Dfrac{\sigma((s,u),(a,w),(t,v))}{\plusvol_{h=k+1}^{k'} \sigma((s_h,u),(a_h,w)) }
\end{split}
\end{equation*}
and we obtain the claim.
\end{proof}

In order to compute the sum of Lemma~\ref{lem:contrib:sum}, we need to iterate over all $s_k$, $-l-1\le k \le -1$ and all $a_{k'}$, $0\le k'\le r$.
For this purpose, we first give an algorithm computing the values of $s_k$.
The algorithm also associates to each $s_k$ a volume of \sfp{}s that will be useful for computing the denominator of the fraction in the sum.

\begin{lemma}
\label{lem:prevlist}
Algorithm~\ref{alg:prevlist} computes the  list $(s_{-1}, f_{-1}), (s_{-2}, f_{-2}), \ldots, (s_{-l-1} = S, f_{-l-1})$
with $s_k$ defined by \latpairset{} and
with $f_k = \sigma((s_{k+1},u), (a_{-1},w))$.
\end{lemma}
\begin{proof}
  The algorithm builds and returns the (initially empty) list $\myR$. The algorithm terminates when $S$ is found.
   Indeed, a return is triggered in three different cases.
   If the empty list is returned at Line~\ref{alg:prevlist:special2},
   this means that
   there exists an $\epsilon > 0$ such that for all $t \in [s-\epsilon, s]$, there is an instantaneous path
   of length $d((s,u),(a,w))$ from $(t,u)$ to $(t,w)$,
   which implies that $S = s$.
   If the return happens after the last value is added to $\myR$  during the {\em for} loop,
   then either $s'<s$, and $s'$ is the largest value such that: $a'-s'<a-s$ (Line~\ref{alg:prevlist:returnlat}); or $a'-s'=a-s$ and $d((s',u),(a',w)) < d((s,u),(a,w))$ (Line~\ref{alg:prevlist:return:dist});
   or there exists $\epsilon >0$ such that for any $t\in [s'-\epsilon,s']$ there is an instantaneous path of length $d((s,u),(a,w))$ from $(t,u)$ to $(t,w)$ (Line~\ref{alg:prevlist:returninf}).
   This corresponds exactly to the definition of $S$.
   Finally if the function returns at Line~\ref{alg:prevlist:append:alpha}, then this means that $S=\alpha$ because none of the above conditions is true for any $s'>\alpha$.

The elements $(s',\myvol)$ added to $\myR$ correspond to all latency pairs $(s',a')$ such that $s' \in \intco{S,s}$, $a'-s' = a-s$, and $d((s',u),(a',w)) = d((s,u),(a,w))$.
These are therefore the $(s_i,a_i)$ in \latpairset\ with $i<0$.
%This ends the proof that the returned list is maximal such that $s_i \in \intco{S,s}$ for all $i \in \intcc{0,c}$.

Let us now show that $\myvol$ contains the desired value when $(s',\myvol)$ is added to $\myR$.

If the empty list is returned at Line~\ref{alg:prevlist:special2}, then this is true. Otherwise, $\myvol$ is initialized to $(0,0)$ and this value is not changed before the first time
the pair $(s',\myvol)$ is appended to $\myR$.
Therefore the first pair appended to $\myR$ is $(s_{-1}, (0,0))$, which is correct
since $f_{-1} = \sigma((s_{0},u), (a_{-1},w)) = (0,0)$
(since there are no paths from $(s_0, u)$ to $(a_{-1},w)$). 

Assume now that the correct value $(s_i, f_i)$ has been added to $\myR$ at one loop iteration,
and that $s_i > S$ (otherwise, as shown above, a return is triggered just after the append and the function returns).
We then have $s'=s_i$ and the value $\sigma((s',u),(a',w))$ is then added to $\myvol$.
$\myvol$ is therefore now equal to $\sigma((s',u),(a',w)) \plusvol{} f_i  
= \sigma((s_i,u),(a_i,w)) \plusvol{} \sigma((s_{i+1},u), (a_{-1},w))
=f_{i-1}$.
Moreover, the loop will skip latency pairs not in \latpairset{} and the next value of $s'$ that will be considered is $s_{i-1}$.
Therefore the next value that is added to $\myR$ is $(s_{i-1}, f_{i-1})$ and
finally all the correct values are added to $\myR$, which completes the proof.
\end{proof}

We also introduce the function {\em NextList} by replacing in {\em PrevList} of Algorithm~\ref{alg:prevlist}:
{\em $d^-_{s}$} by {\em $d^+_{a}$} in Line~\ref{alg:prevlist:special1}; {\em $s$} by {\em $a$} in Line~\ref{alg:prevlist:special2};
 {\em $s'<s$ in $\myLL$ backwards} by {\em  $a'>a$ in $\myLL$ forwards} in Line~\ref{alg:prevlist:for}; all $(s',\myvol)$ appended to $\myR$ by $(a',\myvol)$;
$d^-_{s'}$ by $d^+_{a'}$ in Line~\ref{alg:prevlist:lineeps:2}; and $(\alpha,\myvol)$ by $(\omega,\myvol)$ in the last line.

The obtained function computes the  list $(a_1 ,g_1), (a_{2}, g_{2}), \ldots, (a_{r+1} = A, g_{r+1})$
with $a_k$ defined by \latpairset{} and with $g_k = \sigma((s_1,u), (a_{k-1},w)) $.

\begin{algorithm}
\Fn{Contribution}{
  \KwIn{a link stream $L = (T,V,E)$, $u \in V$, $w \in V$, $(t,v) \in T\times V$, and the latency list $\myLL$ from $u$ to $w$}
  \KwOut{the contribution $\int_{T\times T} \sigma((i,u),(j,w),(t,v))\divvol{} \sigma((i,u),(j,w)) \diff i \diff j$}
  $\myvol\_tv \gets (0,0)$\\
  \For{$(x,y)$ in $\myLL$ \label{alg:contrib:for}}{
    \If{$t \in \intcc{x,y}$ and $(x,u) \reaches (t,v)$ and $(t,v) \reaches (y,w)$}{
      \If{$d((x,u),(y,w)) = d((x,u),(t,v)) + d((t,v),(y,w))$}{
        $\myvol\_tv \gets VSP(L, (x,u), (t,v)) \timesvol VSP(L, (t,v), (y,w))$
      }
      set $(s,a)$ to $(x,y)$ and exit the loop
    }
  }
  \If{$\myvol\_tv = (0,0)$}{
    \Return{$0$}
  }
  $middle \gets \mbox{VSP}(L, (s,u), (a,w))$ \label{alg:contrib:volsa}\\
  $\myPrev \gets \mbox{PrevList}(L, u, w, s, a, \myLL)$\\
  $\myNext \gets \mbox{NextList}(L, u, w, s, a, \myLL)$\\
  $contrib \gets 0$\\
  $s' \gets s$\\
  \For{$(s\_left, left)$ in $\myPrev$ \label{line:begforprevlist}}{
    $a' \gets a$\\
    \For{$(a\_right, right)$ in $\myNext$}{
      $contrib \gets contrib + (s'-s\_left)(a\_right-a')\cdot \myvol\_tv \divvol (left \plusvol right \plusvol middle)$\\
      $a' \gets a\_right$
    }
    $s' \gets s\_left$ \label{line:endforprevlist}
  }
  \Return{$contrib$}
}
\caption{Contribution of two given nodes to the betweenness of a given temporal node}
\label{alg:contrib}
\end{algorithm}

We finally reach the objective of this section.

\begin{theorem}
\label{theo:contrib}
Given a link stream $L=(T,V,E)$, a temporal node $(t,v)$ in $T\times V$, and two nodes $u$ and $w$ in $V$, Algorithm~\ref{alg:contrib} computes the contribution of $u$ and $w$ to the betweenness of $(t,v)$, \ie\ 
$
C_{tv}(u,w) = \int_{i\in T,j\in T} \Dfrac{\sigma((i,u),(j,w),(t,v))}{ \sigma((i,u),(j,w))} \diff i \diff j. 
$
\end{theorem}
\begin{proof}
  If there exists a latency pair $(s,a)$ with non-zero contribution,
  then, for any $i\le s$ and any $j\ge a$, we have:
  $\sigma((i,u),(j,w),(t,v)) = \sigma((s,u),(a,w),(t,v))$.
  The {\em for} loop of Line~\ref{alg:contrib:for} computes $\sigma((s,u),(a,w),(t,w))$ and stores it in $\myvol\_tv$.
  Indeed, since $(t,v)$ is involved in a shortest fastest path from $(s,u)$ to $(a,w)$, necessarily $(s,a)$ is such that $t\in \intcc{s,a}$,
  $(s,u) \reaches (t,v)$, $(t,v) \reaches (a,w)$, and $d((s,u),(a,w)) = d((s,u),(t,v)) + d((t,v),(a,w))$.
  Since there is at most one such pair satisfying the first three conditions,
  the algorithm breaks out of the {\em for} loop if one is found.
If no such latency pair is found, $\myvol\_tv$ is equal to $(0,0)$ at the end of the loop
  and the Algorithm returns $0$.
Notice that, in the special case where $t$ is not an event time and $(t,u) \reaches (t,w)$, then the arguments above do not apply, but the algorithm still returns the correct value: $(t,t)$ is a latency pair that does not belong to the latency list, and the contribution of $(u,w)$ is $0$, which is the returned value.

Remember that $\latpairset{} = (s_{-l},a_{-l}), (s_{-l+1}, a_{-l+1})), \dots, (s_0=s, a_0=a) ,\dots, (s_r,a_r)$ and that $s_{-l-1}=S$ and $a_{r+1}= A$.
  {\em PrevList} (Algorithm~\ref{alg:prevlist}) then computes, according to Lemma~\ref{lem:prevlist},
  the  list $\myPrev = ((s_{-1}, f_{-1}), (s_{-2}, f_{-2}), \ldots, (s_{-l-1} = S, f_{-l-1})$
with $s_k$ defined by \latpairset{} and with $f_k = \sigma((s_{k+1},u), (a_{-1},w))$;
its dual algorithm {\em NextList} computes the list
$\myNext = (a_1 ,g_1), (a_{2}, g_{2}), \ldots, (a_{r+1} = A, g_{r+1})$
with $a_k$ defined by \latpairset{} and with $g_k = \sigma((s_1,u), (a_{k-1},w))$.

According to Lemma~\ref{lem:contrib:sum},
$$
C_{tv}(u,w) =
  \sum_{k=-l-1}^{-1} \sum_{k'=0}^r (s_{k+1} - s_k ) \cdot (a_{k'+1} - a_{k'} ) \cdot
  \Dfrac{\sigma((s,u),(a,w),(t,v))}{\plusvol_{h=k+1}^{k'} \sigma((s_h,u),(a_h,w))}.
$$

Lines~\ref{line:begforprevlist} to~\ref{line:endforprevlist} of Algorithm~\ref{alg:contrib} compute this sum.
First notice that $s'$ is initialized to $s = s_0$
and $s\_left$ loops over values in $\myPrev$, starting with $s_{-1}$.
At the end of each iteration $s'$ is set to $s\_left$ and therefore $s'$ and $s\_left$ loop
over all consecutive values $s_{k+1}, s_k$ for $s_k$ in $\myPrev$.
The value of  $left$ is $f_k = \sigma((s_{k+1},u), (a_{-1},w))
=\plusvol_{h=k+1}^{-1} \sigma((s_h,u),(a_h,w))$, as explained in the characterization of \myPrev\ above.

Similarly, in the inner {\em for} loop $a'$ and $a\_right$ loop over all values
$a_{k'}, a_{k'+1}$ for $a_{k'+1}$ in $\myNext$,
and $right = g_{k'+1} = \sigma((s_1,u), (a_{k'},w))
=\plusvol_{h=1}^{k'} \sigma((s_h,u),(a_h,w))$.

The value of $middle$ has been set to  $\sigma((s,u),(a,w))$ (Line~\ref{alg:contrib:volsa}).
Therefore $left \plusvol right \plusvol middle =
\plusvol_{h=k+1}^{k'} \sigma((s_h,u),(a_h,w))$.
Finally, one iteration of the inner loop adds the value
$(s_{k+1} - s_k)(a_{k'+1} - a_{k'}) \sigma((s,u),(a,w)) \divvol \plusvol_{h=k+1}^{k'} \sigma((s_h,u),(a_h,w))$,
which is exactly one term of the sum in Equation~\ref{eq:sumcontrib}, and the loop itself ensures we obtain the whole sum.
\end{proof}

%%%%%%%%%%%%%%%%%%%%%%%%%%%%%%%%%%%%%%%%
\section{Betweenness of a temporal node}
\label{sec:betweenness}

We now have all needed building blocks for computing the betweenness of any given temporal node: we just have to sum the contribution of each node pair, see Algorithm~\ref{alg:bet}.

\begin{algorithm}[!h]
\Fn{Betweenness}{
 \KwIn{a link stream $L=(T,V,E)$ and $(t,v)\in T\times V$}
 \KwOut{the betweenness of $(t,v)$}
 $B \gets 0$\\
 \For{$u \in V$}{
   $\myLL \gets \mbox{Latency-lists}(u)$\\
   \For{$w \in V$}{
     $B \gets B + \mbox{Contribution}(L, u, w, (t,v), \myLL [w])$
   }      
 }
 \Return{B}
}
\caption{Betweenness of a temporal node}
\label{alg:bet}
\end{algorithm}

The key variables describing the size of our algorithm inputs, for a given link stream $L=(T,V,E)$, are the number of nodes $n = |V|$ and the number of link segments $\overline{m}$, {\em i.e.} the number of maximal intervals in $E$. Notice that the number of event times $|\eventtimes|$ is at most $2\cdot \overline{m}$, and so it is in $O(\overline{m})$. Likewise, the number of links $m_t = |\{uv, (t,uv)\in E\}|$ at time $t$, for any $t$, as well as the number of links $m = |\{uv, \exists t, (t,uv)\in E\}|$ in the induced graph, also are in $O(\overline{m})$. Then, the complexity of all algorithms presented in this paper is clearly polynomial in $n$ and $\overline{m}$, which makes Algorithm~\ref{alg:bet} polynomial itself.

We display in Figure~\ref{fig:ex-intro-result} the results obtained in the case of Figure~\ref{fig:ex-intro}, where we computed the betweenness of more than $5000$ temporal nodes in a few seconds. We provide the implementation at \cite{btwurl}.

\begin{figure}[!h]
\centering
\includegraphics[width=.8\textwidth]{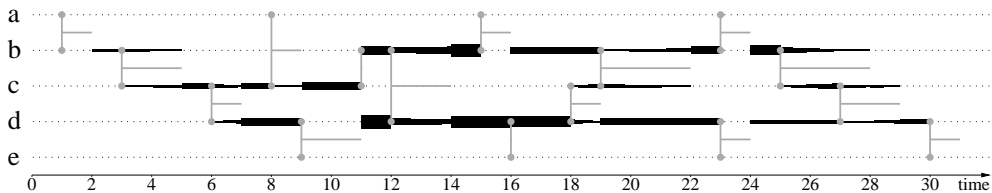}
\caption{
Results of betweenness computations on the example of Figure~\ref{fig:ex-intro}. We computed the betweenness of $(t,v)$ for all $v$ and $t$ equal to $\alpha + i \cdot \frac{\omega - \alpha}{1000}$, for $i = 0..1000$. The obtained value is displayed at $(t,v)$ as a black rectangle of width $\frac{\omega-\alpha}{1000}$ and height propotional to the betweenness of $(t,v)$. Dotted lines represent betweenness values equal to $0$.
}
\label{fig:ex-intro-result}
\end{figure}

%%%%%%%%%%%%%%%%%%%%%%%%%%%%%%%%%%%%%%%%
\section{Related work}
\label{sec:related}

Betweenness computations are first related to {\bf path computations}. Temporal paths already received much attention, in particular optimal path computations according to several criteria (like length, duration, and/or arrival time), see for instance~\cite{BuiXuan2003Computing,Wu2014Path,Wu2016Reachability}. However, most of these works are limited to discrete time and instantaneous links; only few consider continuous time and links with duration~\cite{Whitbeck2012Temporal,Yuan2019Constrained,DBLP:conf/asunam/Simard19}. Then, the focus is on finding optimal paths or computing distances and latencies \cite{DBLP:conf/asunam/Simard19}, not counting them as we do here. The authors of~\cite{AfrasiabiRad2017Computation} notice that the number of foremost paths (temporal paths with minimal arrival time) may be exponential. The problem that we consider here is quite different because we handle continuous time and links with duration. This leads to the concept of finite volumes of uncountable path sets, that never appeared in previous literature, up to our knowledge.

The graph betweenness itself also has been studied in temporal settings. A first line of study focuses on {\bf updating betweenness values} upon link arrival or departure, see for instance~\cite{Bergami2017Faster,Green2012Fast}. This is quite different from our work: the considered paths are classical (static) graph paths, and the considered betweenness is the classical one, at each time instant.

Several works consider {\bf temporal betweenness extensions} that rely on various kinds of optimal (fastest, shortests, foremost, etc) paths. Most have a node-centric view: they define a value for each node, not for each temporal node, see for instance \cite{Habiba2011Betweenness,Tang2010Analyzing,Nicosia2013Graph,Williams2016Spatio,Kim2012Temporal}.
Others define a value for each temporal node, like in our case. For instance, \cite{Takaguchi2015Coverage} proposes {\em coverage centrality} of $(t,v)$, defined as the fraction of pairs of (non-temporal) nodes for which there exists a fastest path involving $(t,v)$. Bu\ss{} et al.~\cite{Buss2020Algorithmic} consider instantaneous links and define betweenness centralities for various types of optimal temporal paths.
The authors of~\cite{Tang2010Analyzing,Gunturi2017Scalable} define a betweenness value for each temporal node, based on foremost paths or other optimal paths. The algorithm in \cite{Gunturi2017Scalable} starts by identifying time instants for which foremost path trees are stable, which is related to our latency pairs. In \cite{Tsalouchiou2019Temporal}, the authors combine the length and duration of paths using a tunable parameter, and focus on instantaneous links.

All these works assume discrete time, which implies finite sets of shortest paths. Instead, we consider continuous time, leading to uncountable sets of paths, with finite volume. In addition, these works keep a partly node-centric point of view by considering paths between nodes; we push the integration of temporal aspects further by considering paths between temporal nodes. This makes an important difference, since the node-centric view misses locally-optimal paths: they only count paths with a given duration or length between pairs of nodes (for any starting and arrival times), whereas our approach combines a variety of locally shortest fastest paths, with different durations and lengths. This raises different algorithmic challenges, like the computation of latency lists and the selection of appropriate contributing latency pairs.

{\bf Closer to our work}, \cite{AfrasiabiRad2017Computation} and \cite{Pereira2016Evolving} consider optimal paths within time slices, thus obtaining a betweenness value for each node for each time slice. Again, they only consider discrete time, and only a limited number of source and target temporal nodes.

Finally, the generalized betweenness that we consider in this paper, by dealing with continuous time, links with or without duration, as well as paths between all pairs of temporal nodes, raises original algorithmic questions that are not present in previous literature.

%\input{related_cm.tex}

%%%%%%%%%%%%%%%%%%%%%%%%%%%%%%%%%%%%%%%%
\section{Conclusion}

We presented the first algorithms to compute betweenness centrality of temporal nodes in link streams. To obtain these algorithms, we identified and addressed several original challenges, like the definition and computation of volumes of infinite sets of paths, the computation of all latency pairs from any node to all others, or the transformation of continuous-time integrals into discrete sums over finite numbers of time intervals. Each of these building blocks has its own interest, in particular the computation of shortest path volumes from a given temporal node. The complexity of obtained algorithms is polynomial in time and space, and we provide an implementation in python \cite{btwurl}.

%However, this is a first step only. Indeed, many directions remain to explore.

Our algorithm leaves room for complexity improvement. In particular, it seems promising to explore extensions to link streams of approaches like Brandes' for betweenness on graphs \cite{Brandes2001Faster}. Another important direction is to design algorithms to compute the betweeness of {\em all} temporal nodes rather than just one: iterating our algorithm over many temporal nodes leads to much redundancy. However, keep in mind that there is an infinite number of temporal nodes; one may then try to infer the betweenness of any of them from the betweenness of a finite number of them, for instance each node at each event time. This seems non-trivial, though, and an open question.

Going further, one may try to design approximate algorithms. Indeed, the best known time complexity of betweenness computations in graphs is $O(n m)$ \cite{Brandes2001Faster} and it cannot be lower in link streams, since graphs are special cases \cite{DBLP:journals/snam/LatapyVM18}. This is prohibitive in many practical cases, leading to much work on approximate computations, that typically compute shortest paths from some nodes only \cite{10.1145/2556195.2556224,DBLP:conf/ipps/GrintenM20}. Such approaches are very relevant in link streams too, where the contribution of only a few node pairs may give reasonably accurate approximates, at a much lower cost than exact computations. This remains to explore, though.

An even more challenging direction is to embrace the streaming nature of link streams, and design on-line and/or streaming algorithms for betweenness. Such algorithms do not store the data in memory; they compute results on-the-fly and output them as soon as they are available. They would be of high theoretical and practical interest, but they raise many challenges.

Another interesting family of perspectives consists in extending or restricting the considered input. In particular, one may consider stream graphs instead of link streams: in stream graphs, nodes are not always present, leading to more subtle path, distance, and latency concepts \cite{DBLP:journals/snam/LatapyVM18}. We considered here streams with link (and node) presence times equal to unions of disjoint closed intervals (including singletons); another extension would be to consider more general cases, like for instance unions of disjoint closed or open intervals. Also, weighted and/or directed stream graphs and link streams \cite{Latapy2019} lead to more complex concepts of shortest fastest paths, and our definitions of volumes may be extended to these cases. Conversely, one may consider more specific situations, like discrete time streams, or link stream with instantaneous links only. Such cases often appear in practice, and it may be possible to design more efficient algorithms for them.

Extending our algorithms to variants of the betweenness concept itself also is an interesting perspective. One may for instance consider betweenness of {\em links} rather nodes, or consider paths of other kinds than shortest fastest ones, {\em e.g.} foremost ones \cite{DBLP:journals/snam/LatapyVM18}

Finally, this paper opens the perspective of practical uses of betweenness in link streams, since until now only the definition was available. It is now possible to explore how betweenness is distributed in (small scale) real-world cases, and gain insight from this. It may also be used to extend important graph algorithms to link streams, like the computation of communities by iteratively removing temporal nodes of highest betweenness, in a way similar to \cite{Girvan2002Community} that iteratively removes links of highest betweenness.

\medskip
\noindent
{\bf Acknowledgements.}
This work is funded in part by the ANR (French National Agency of Research) under the Limass project (ANR-19-CE23-0010) and the FiT LabCom grant.

\bibliographystyle{plain}
\bibliography{biblio}

\end{document}